\pdfoutput=1
\documentclass[final,1p,times]{elsarticle}

\usepackage{lineno,hyperref,latexsym,amsmath,amsthm,amsfonts, bookmark,color}
\usepackage{listings}

\newtheorem{lem}{Lemma}%[section]
\newtheorem{thm}{Theorem}%[section]
\newtheorem{coro}{Corollary}

\newtheorem{prop}{Proposition}
\newtheorem{exa}{Example}

\begin{document}
\newenvironment{prf}[1][Proof]{\noindent\textit{#1}\quad }

\begin{frontmatter}
\title{A tight upper bound on the number of  non-zero weights of   a constacyclic  code}

\author[1]{Hanglong Zhang\corref{cor1}}
\ead{zhanghl9451@163.com}
\author[1,2]{Xiwang Cao}
\ead{xwcao@nuaa.edu.cn}
{\cortext[cor1]{Corresponding author.}}

\tnotetext[]
{This research is supported by National Natural Science Foundation  of China under Grant 11771007.}

%\author[3]{Xiang-Dong Hou}
%\ead{xhou@usf.edu}
%\author[1]{Jiafu Mi}
%\ead{mijiafu@163.com}
%{\cortext[cor1]{Corresponding author.}}
%\author[1,4]{Shanding Xu}
%\ead{sdxzx11@163.com}

\address[1]{Department of Mathematics, Nanjing University of Aeronautics and Astronautics, Nanjing 210016, China}
\address[2]{Key Laboratory of Mathematical Modelling and High Performance Computing of Air Vehicles (NUAA), MIIT,  Nanjing  211106,  China}

\begin{abstract}
For a  simple-root $\lambda$-constacyclic code $\mathcal{C}$ over $\mathbb{F}_q$,	let $\langle\rho\rangle$ and $\langle\rho,M\rangle$ be the subgroups of the automorphism group of $\mathcal{C}$ generated by the cyclic shift $\rho$, and by  the cyclic shift $\rho$ and the scalar multiplication $M$, respectively. Let $N_G(\mathcal{C}^\ast)$ be the number of orbits of a subgroup $G$ of automorphism group of $\mathcal{C}$ acting  on $\mathcal{C}^\ast=\mathcal{C}\backslash\{0\}$.
 In this paper,  we establish  explicit formulas for $N_{\langle\rho\rangle}(\mathcal{C}^\ast)$ and $N_{\langle\rho,M\rangle}(\mathcal{C}^\ast)$.  Consequently, we  derive a upper bound on the number of nonzero weights of   $\mathcal{C}$. We present some irreducible and reducible  $\lambda$-constacyclic codes, which show that the upper bound is tight. A sufficient condition to  guarantee $N_{\langle\rho\rangle}(\mathcal{C}^\ast)=N_{\langle\rho,M\rangle}(\mathcal{C}^\ast)$ is presented. 
\end{abstract}

\end{frontmatter}

\section{Introduction}

In this paper, let $n$ be  a positive integer, $q$  a prime power and $\gcd(n,q)=1$. Let  $\mathbb{F}_q$ be the finite field with $q$ elements. An $[n,k]$ linear code $\mathcal{C}$ over $\mathbb{F}_q$ is a $k$-dimensional subspace of $\mathbb{F}^n_q$. For a codeword $c=(c_0,c_1,\cdots,c_{n-1})\in\mathcal{C}$, the Hamming weight of $c$ is the number of nonzero coordinates in $c$. The minimum Hamming weight of $\mathcal{C}$ is the smallest Hamming weight of the nonzero codewords of $\mathcal{C}$. Let $A_i$ be the number of codewords of weight $i$ in $\mathcal{C}$. The $(n+1)$-tuple $(A_0,A_1,\cdots,A_{n})$ is called the weight distribution or weight spectrum of $\mathcal{C}$.
For $\lambda\in\mathbb{F}_q^{\ast}$, a linear  code $\mathcal{C}$ of length $n$ over $\mathbb{F}_q$ is called $\lambda$-constacyclic  if for every $(c_0,c_1,...,c_{n-1})\in\mathcal{C}$ implies $(\lambda c_{n-1},c_0,c_1,...,c_{n-2})\in\mathcal{C}$. A  $1$-constacyclic code is called a cyclic code.

 Delsarte  \cite{dels} proved that the number of nonzero weights of a linear code  plays an important role in the research and application of coding theory. In \cite{dels}, the minimum Hamming weight and the number of nonzero weights of a given code and its dual code were called ``four fundamental parameters of a code". Let $\theta$ and $\theta'$ be the numbers of nonzero weights  of  a linear code $\mathcal{C}$ of length $n$ over $\mathbb{F}_q$ and that of its dual, respectively.
  Delsarte showed that  $\theta$ and $\theta'$ can be used to estimate the   size of the code:
  $$  q^n \big/ \bigg(  \sum_{i=0}^{\theta'} \binom{n}{i}(q-1)^i \bigg) \leq \big|\mathcal{C}\big| \leq  \sum_{j=0}^{\theta} \binom{n}{j}(q-1)^j.$$
 Some close relations  between codes and combinatorial designs based on the number of nonzero weights of  codes were showed  in \cite{assm} and \cite{dels}.

 The number of nonzero weights of  linear codes is a useful parameter, however, it is very hard to get an explicit formula  for this number in general.  Thus many scholars tried to get some lower and upper bounds on the number of nonzero weights of  linear codes. Recently, Shi \emph{et al.} released a series of papers \cite{shi3,shi1,shi2} to study the number of nonzero weights of linear codes.
  In \cite{shi3}, Shi  \emph{et al.}
   conjectured that the largest number of nonzero weights of a linear code of dimension $k$ over $\mathbb{F}_q$  is  $\frac{q^k-1}{q-1}$, and proved that the bound is tight for $q=2$ or $k=2$. The conjecture was completely proved by Alderson and  Neri \cite{alde}.  
   In \cite{alde0},  Alderson defined full weight codes as the codes of length $n$ with exactly $n$ nonzero weights, and provided some sufficient and  necessary conditions for the existence of full weight codes.
    In \cite{shi1},  Shi  \emph{et al.}  presented lower and upper bounds on  the largest number of nonzero weights of cyclic codes;   sharper upper bounds  were given for some special classes of cyclic codes; the numbers of nonzero weights of $q$-ary Reed-Muller codes, and of many $q$-ary Hamming codes were determined.  
 Later, in \cite{shi2}, Shi \emph{et al.}  investigated the largest number of nonzero weights of quasi-cyclic codes. Several lower and upper bounds were provided.
 In \cite{chen1}, Chen and Zhang studied the group action of a subgroup of automorphism group of a cyclic code $\mathcal{C}$ generated by the cyclic shift and the scalar multiplication on $\mathcal{C}\backslash\{0\}$. An  explicit formula for the number of  orbits of the group action was given. Consequently, a upper bound on the number of nonzero weights of $\mathcal{C}$ was derived. A sufficient and  necessary condition for codes meeting  the bound was  exhibited. Some irreducible and reducible cyclic codes  meeting the bound were presented, revealing that the bound is tight.
 
 Inspired by \cite{chen1}, in this paper, we aim to  give a tight upper bound on the number of nonzero weights of a simple-root $\lambda$-constacyclic code $\mathcal{C}$ over $\mathbb{F}_q$. 
 In \cite[Proposition \uppercase\expandafter{\romannumeral2}.2]{chen1}, the author observed that the number of nonzero weights of a linear code is bounded from above by  the number of orbits of  the subgroup of automorphism group  acting  on the linear code. By the virtue of this proposition, our
goal is  turn to determine the number of  orbits   of  the subgroup of automorphism group  acting  on a $\lambda$-constacyclic code.
   We give    explicit formulas for   $N_{\langle\rho\rangle}(\mathcal{C}^\ast)$ and $N_{\langle\rho,M\rangle}(\mathcal{C}^\ast)$.  Consequently, by using  \cite[Proposition \uppercase\expandafter{\romannumeral2}.2]{chen1}, we  derive a upper bound on the number of nonzero weights of   $\mathcal{C}$.  We exhibit some irreducible and reducible $\lambda$-constacyclic codes meeting this bound, proving that this bound is tight.  A sufficient condition to  guarantee  $N_{\langle\rho\rangle}(\mathcal{C}^\ast)=$ $N_{\langle\rho,M\rangle}(\mathcal{C}^\ast)$ is presented. Our main results  generalize some of the results in \cite{chen1}.
 
 Here goes  the structure of this paper. In Section 2, we introduce some definitions and results of group actions, $\lambda$-constacyclic codes and their generating idempotents. 
 Section 3 is divided  into two subsections 3.1 and 3.2, which investigate $N_{\langle\rho\rangle}(\mathcal{C}^\ast)$ and $N_{\langle\rho,M\rangle}(\mathcal{C}^\ast)$, respectively. This leads to our main conclusion, Theorem \ref{thm2}.
 Some examples prove that  the bound is tight.  A sufficient condition to  guarantee $N_{\langle\rho\rangle}(\mathcal{C}^\ast)=$  $N_{\langle\rho,M\rangle}(\mathcal{C}^\ast)$ is presented.

%For $\lambda\in\mathbb{F}_q^{\ast}$, an $[n,k]$ code $\mathcal{C}$ over $\mathbb{F}_q$ is called $\lambda$-constacyclic  if for every $(c_0,c_1,...,c_{n-1})\in\mathcal{C}$ implies $(\lambda c_{n-1},c_0,c_1,...,c_{n-2})\in\mathcal{C}$. 
%Since the vector $(c_0,c_1,...,c_{n-1})\in\mathbb{F}^n_q$ can be  identified with a  polynomial $c_0+c_1x+...+c_{n-1}x^{n-1}\in \mathbb{F}_q[x]$,  a $\lambda$-constacyclic  code over $\mathbb{F}_q$ is an ideal of the quotient ring $\mathcal{R}_{n,\lambda}^{(q)}:=\mathbb{F}_q[x]/ \langle x^n-\lambda \rangle$, where $ \langle x^n-\lambda \rangle$ is the  ideal of $\mathbb{F}_q[x]$ generated by $x^n-\lambda$. Furthermore, every ideal of $\mathcal{R}_{n,\lambda}^{(q)}$ is a principal ideal, so $\mathcal{C}$ is generated by a monic divisor $g(x)$ of $x^n-\lambda$.   The  quotient ring  $\mathcal{R}_{n,\lambda}^{(q)}$ is semi-simple when $\gcd(n,q)=1$. In this paper, we always assume that $\gcd(n,q)=1$. Therefore each constacyclic code in  $\mathcal{R}_{n,\lambda}^{(q)}$ can be   generated by a unique idempotent. This idempotent is called the generating idempotent of the constacyclic code.

\section{Preliminaries}
\subsection{\texorpdfstring {$\lambda$}{}-constacyclic code and primitive idempotents}

Since the vector $(c_0,c_1,...,c_{n-1})\in\mathbb{F}^n_q$ can be  identified with a  polynomial $c_0+c_1x+...+c_{n-1}x^{n-1}\in\mathbb{F}_q[x]/ \langle x^n-\lambda \rangle$,  a $\lambda$-constacyclic  code of length $n$ over $\mathbb{F}_q$ is an ideal of the quotient ring $\mathcal{R}_{n,\lambda}^{(q)}:=\mathbb{F}_q[x]/ \langle x^n-\lambda \rangle$. Furthermore, every ideal in  $\mathcal{R}_{n,\lambda}^{(q)}$ is  principal, then $\mathcal{C}$ can be expressed as $\mathcal{C}=\langle g(x)\rangle$, where $g(x)$ is a divisor of $x^n-\lambda$ and has the least degree. The  $g(x)$ is unique and called the generator polynomial of $\mathcal{C}$ and $h(x)=\frac{x^n-\lambda}{g(x)}$ is called the check polynomial of $\mathcal{C}$.  The  quotient ring  $\mathcal{R}_{n,\lambda}^{(q)}$ is semi-simple when $\gcd(n,q)=1$. Each $\lambda$-constacyclic code in  $\mathcal{R}_{n,\lambda}^{(q)}$ can be   generated by a unique idempotent. This idempotent is called the generating idempotent of the $\lambda$-constacyclic code.

Let $\lambda\in\mathbb{F}_q^{\ast}$ with multiplicative order $t$. Since $t\mid q-1$ and $\gcd(n,q)=1$, we have $\gcd(tn,q)=1$.
 Let $m$ be the least integer such that $tn \mid q^m-1$.  Let $\omega$ be a primitive element of $\mathbb{F}_{q^m}^{\ast}$, $\lambda=\omega^{\frac{q^m-1}{t}}$ and $\zeta=\omega^{\frac{q^m-1}{tn}}$. Then $\zeta$ is a primitive $tn$-th root of unity in $\mathbb{F}_{q^m}^{\ast}$ and $\lambda=\zeta^{n}$. Therefore, $x^n-\lambda=\prod_{i=0}^{n-1}(x-\zeta^{1+ti})$. 
%Let 
%\begin{eqnarray*}
%&C_{1+t\alpha_0}=\{1,q,\cdots,q^{d_0-1}\}, C_{1+t\alpha_1}=\{1+t\alpha_1,(1+t\alpha_1)q,\cdots,(1+t\alpha_1)q^{d_1-1}\},\cdots,\\
%&C_{1+t\alpha_s}= \{1+t\alpha_s,(1+t\alpha_s)q,\cdots,(1+t\alpha_s)q^{d_s-1}\}
%\end{eqnarray*}
%Denote $C_{1+ti}=\{1+ti, (1+ti)q, (1+ti)q^{2},\cdots\, (1+ti)q^{d_i-1}\}$ the $q$-cyclotomic coset containing $1+ti$  modulo $tn$ where $d_i$ is the least integer such that $1+ti\equiv (1+ti)q^{d} ~({\rm mod}~ tn)$  for  $i\in\{0,\cdots,n-1\}$.
Let
\begin{align*}
	 C_{1+t\alpha_0}=\{1,q,\cdots,q^{d_0-1}\}, C_{1+t\alpha_1}=\{ 1+t\alpha_1,(1+t\alpha_1)q,\cdots,(1+t\alpha_1)q^{d_1-1}  \},\\
	 \cdots,
	 C_{1+t\alpha_s}=\{ 1+t\alpha_s,(1+t\alpha_s)q,\cdots,(1+t\alpha_s)q^{d_s-1}  \}
\end{align*}
be all distinct $q$-cyclotomic cosets modulo $tn$ contained in the set $\mathcal{S}:=\{1+ti|i=0,\cdots,n-1\}$,  where $0=\alpha_0<\alpha_1<\cdots<\alpha_s\leq n-1$ and  $d_j$ is the smallest positive integer such that $(1+t\alpha_j)q^{d_j}\equiv 1+t\alpha_j ({\rm mod}~ tn)$ for $0\leq j\leq s$.  Then $|C_{1+t\alpha_j}|=d_{j}$  for $0\leq j\leq s$. As is easily  checked, $C_{1+t\alpha_0},\cdots,C_{1+t\alpha_s}$ form a partition of the set $\mathcal{S}$. Hence, $x^n-\lambda$ can be decomposed into
$x^n-\lambda=\prod_{i=0}^{s}m_i(x)$,
where $m_i(x)=\prod_{j\in C_{1+t\alpha_i}}(x-\zeta^{j})$  is irreducible over $\mathbb{F}_q$ and $m_0(x),\cdots,m_s(x)$ are pairwise coprime. For each $0\leq i\leq s$,
let $\mathcal{I}_i$ be the minimal ideal in $\mathcal{R}_{n,\lambda}^{(q)}$ with the generator polynomial $\frac{x^n-\lambda}{m_i(x)}$ and let  $\varepsilon_i$ be the generating idempotent of $\mathcal{I}_i$.  Then $\dim(\mathcal{I}_i)=|C_{1+t\alpha_i}|=d_i$. The idempotents $\varepsilon_0,\varepsilon_1,\cdots,\varepsilon_s$ are called the primitive idempotents of $\mathcal{R}_{n,\lambda}^{(q)}$.   It is known that $\varepsilon_i(\zeta^j)=1$ when $j\in C_{1+t\alpha_i}$ and  $\varepsilon_i(\zeta^k)=0$ when $k\in \mathcal{S} \backslash{}C_{1+t\alpha_i}$. Hence,   $\mathcal{R}_{n,\lambda}^{(q)}$ has $s+1$ primitive idempotents 
\begin{equation}\label{eq2}
 \varepsilon_i=\frac{1}{n} \sum_{v=0}^{n-1}\sum_{h\in C_{1+t\alpha_i}}\zeta^{-vh}x^v \in \mathbb{F}_{q}[x]	 
\end{equation}
 for $0\leq i\leq s$. For the  quotient ring $\mathbb{F}_{q^m}[x]/ \langle x^n-\lambda \rangle$, there has $n$ primitive idempotents
\begin{eqnarray}\label{eq1}
	  e_{1+tj}=\frac{1}{n} \sum_{v=0}^{n-1}\zeta^{-v(1+tj)}x^v \in \mathbb{F}_{q^m}[x]
\end{eqnarray}
for $0\leq j\leq n-1$.  By \cite[Lemma 2.2]{cao1}, $\mathcal{R}_{n,\lambda}^{(q)}$ is the direct sum of $ \mathcal{I}_i$ for $0\leq i\leq s$, in symbols,
$$ \mathcal{R}_{n,\lambda}^{(q)}= \mathcal{I}_0 \bigoplus  \mathcal{I}_1 \bigoplus \cdots \bigoplus  \mathcal{I}_s.  $$
For each $0\leq i\leq s$, since the irreducible $\lambda$-constacyclic code $\mathcal{I}_i$ can be generated by $\varepsilon_i$,  we have 
\begin{eqnarray*}
	 \mathcal{I}_i=\bigg\{~c(x)\varepsilon_i~\big| ~ c(x)=\sum_{h=0}^{d_i-1}c_hx^h ~{\rm and}~ c_h\in\mathbb{F}_q ~{\rm for}~ 0\leq h\leq d_i-1\bigg\}.
\end{eqnarray*}
From  (\ref{eq2}) and  (\ref{eq1}), we have $\varepsilon_i=\sum_{j\in C_{1+t\alpha_i}} e_j$ for  $0\leq i \leq s$, and   $x^v=\sum_{j=0}^{n-1}\zeta^{v(1+tj)}e_{1+tj}$ for $0\leq v\leq n-1$. Hence, 
\begin{eqnarray}\label{main}
	\mathcal{I}_i=\bigg\{\sum_{j=0}^{d_i-1}(\sum_{h=0}^{d_i-1} c_h\zeta^{h(1+t\alpha_i)q^j})e_{(1+t\alpha_i)q^j}~ \bigg|~ c_h\in\mathbb{F}_q ~{\rm for}~ 0\leq h\leq  d_i-1 \bigg\}.
\end{eqnarray}
Note that for each $0\leq i\leq s$,  the check polynomial  of $\mathcal{I}_i$  is 
	$m_i(x)=\prod_{k\in C_{1+t\alpha_i}}(x-\zeta^{j})$, 
and then  $\mathcal{I}_i$ corresponds to the $q$-cyclotomic coset $	 C_{1+t\alpha_i}=\{1+t\alpha_i, (1+t\alpha_i)q,\cdots,(1+t\alpha_i)q^{d_i-1}\}$.

\subsection{Automorphism group and group actions}
A permutation matrix is a square matrix with exactly one 1 in each row and column and 0s elsewhere. The set of coordinate permutations that map a code $\mathcal{C}$ to itself is  a group, which is referred to as the permutation automorphism group of $\mathcal{C}$ and denoted by ${\rm PAut}(\mathcal{C})$. 

A monomial matrix over $\mathbb{F}_q$ is   a square matrix with exactly one nonzero element of $\mathbb{F}_q$ in each row and column and 0s elsewhere. A monomial matrix $M$ can be written  either in the forms of $PD$ or $D'P$ with a permutation  matrix $P$, and $D$ and $D'$ being diagonal matrices.  The set of monomial matrices that map  a code $\mathcal{C}$ to itself is  a group, which is called  the monomial automorphism group of $\mathcal{C}$ and denoted by ${\rm MAut}(\mathcal{C})$.

 The automorphism group ${\rm Aut}\mathcal{(C)}$ is a set of maps with the  form  of $M\sigma$ that map a code $\mathcal{C}$ to itself, where $M$ is a monomial matrix and $\sigma$ is a field automorphism of $\mathbb{F}_q$. Then we have ${\rm PAut}\mathcal{(C)}\subseteq{\rm MAut}\mathcal{(C)}\subseteq{\rm Aut}\mathcal{(C)}$. All the above groups are the same if the code is binary.

Next, we recall the knowledge about  group actions.	Suppose that a group $G$ acts on a nonempty finite  set $X$. For $x\in X$, $Gx=\{~gx~|g\in G \}$   is called the orbit  of this group action containing $x$ (for short, $G$-orbit). The set of all $G$-orbits is denoted as 
\begin{equation}
	 G \backslash{} X=\{Gx|x\in X\}.
\end{equation}
Evidently, the distinct $G$-orbits form a partition of $X$. Suppose $\mathcal{G}$ is a subgroup of ${\rm Aut}\mathcal{(C)}$, which naturally gives a rise to a group action of $\mathcal{G}$ on $\mathcal{C}^{\ast}=\mathcal{C}\backslash \{0\}$.
The following proposition presented in \cite{chen1} shows that the number of orbits of  $\mathcal{G}$ on  $\mathcal{C}^\ast$ is a upper bound for the number of nonzero weights of $\mathcal{C}$, with equality if and only if any two codewords with the same weight are in the same orbit.

\begin{prop}\cite[Proposition \uppercase\expandafter{\romannumeral2}.2]{chen1}\label{pro1}
	 Let $\mathcal{C}$ be a linear code of length $n$ over $\mathbb{F}_q$
	 with $l$ nonzero weights and let ${\rm Aut}\mathcal{(C)}$ be the automorphism
	 group of $\mathcal{C}$. Suppose that $\mathcal{G}$ is a subgroup of ${\rm Aut}\mathcal{(C)}$. If the
	 number of orbits of $\mathcal{G}$ on $\mathcal{C}^{\ast}=\mathcal{C}\backslash \{0\}$ is equal
	 to $N$, then $l\leq N$. Moreover, the equality holds if and only if for any two
	 non-zero codewords $c_1,c_2\in\mathcal{C}$ with the same weight, there
	 exists an automorphism $A\in \mathcal{G}$ such that $c_1A=c_2$.
\end{prop}
\begin{lem}(Burnside's lemma)\cite[Theorem 2.113]{rotman}\label{lem1}
	 Let $G$ act  on a  finite set $X$. Then the number of orbits of  $G$ on $X$ is equal to 
	 \begin{equation*}
	 	 \frac{1}{|G|} \sum_{g\in G} \big| {\rm Fix}(g) \big|,
	 \end{equation*}
 where ${\rm Fix}(g)=\{x\in X| gx=x\}$.
\end{lem}

\begin{lem}\cite[Lemma \uppercase\expandafter{\romannumeral2}.3]{chen1}\label{lem4}
	 Let $G$ be  a finite group acting on a nonzero  finite  set $X$ and let $N$ be a normal subgroup of $G$. It is clear that $N$ naturally acts on $X$. Suppose the set of $N$-orbits is denoted by $N \backslash{} X=\{ Nx|x\in X\}$. Then the quotient group $G/N$ acts on $N \backslash{} X$ and the number of orbits of $G$ on $X$ is equal to the number of orbits  of  $G/N$  on $N \backslash{} X$.
\end{lem}

We define two  $\mathbb{F}_q $-linear maps on $ \mathcal{R}_{n,\lambda}^{(q)}$,  cyclic shift $\rho$ and scalar multiplication $\sigma_b$, respectively:
\begin{equation*}
	 \rho:  \mathcal{R}_{n,\lambda}^{(q)}\longrightarrow \mathcal{R}_{n,\lambda}^{(q)}, \\ 
	 \sum_{i=0}^{n-1}a_ix^i \longrightarrow  \rho(\sum_{i=0}^{n-1}a_ix^i)=\sum_{i=0}^{n-1}a_ix^{i+1},
\end{equation*}
and for $b\in \mathbb{F}_q^{\ast}$,
\begin{equation*}
	\sigma_b:  \mathcal{R}_{n,\lambda}^{(q)}\longrightarrow \mathcal{R}_{n,\lambda}^{(q)}, \\ 
	\sum_{i=0}^{n-1}a_ix^i \longrightarrow  \sigma_b(\sum_{i=0}^{n-1}a_ix^i)=\sum_{i=0}^{n-1}ba_ix^{i}.
\end{equation*}
Evidently,   $\rho$ and $\sigma_b$ are  $\mathbb{F}_q$-vector space automorphisms of $\mathcal{R}_{n,\lambda}^{(q)}$ and they both  belong to ${\rm Aut}(\mathcal{C})$ for any $\lambda$-constacyclic code $\mathcal{C}$ over $\mathbb{F}_q$.
For any codeword $c=(c_0,c_1,\cdots,c_{n-1})$ in a $\lambda$-constacyclic code of length $n$ over $\mathbb{F}_q$, we have $\rho(c)=(\lambda c_{n-1},c_0,c_1,\cdots,c_{n-2})$ and $\rho^{tn}(c)=(\lambda ^t c_0,\lambda ^tc_1,\cdots,$ $\lambda ^tc_{n-1})=c$ since $t$ is the order of $\lambda\in \mathbb{F}_q^{\ast}$. Hence, $\langle \rho \rangle$ is a subgroup of ${\rm Aut}(\mathcal{C})$ generated by  $\rho$ with  order $tn$.
 Denote  $M:=\{\sigma_b|b\in\mathbb{F}_q^{\ast}\}$. Then $M$ is a subgroup of  ${\rm Aut}\mathcal{(C)}$ with order $q-1$. Let     $\langle \rho, M\rangle$ denote the   subgroup of  ${\rm Aut}\mathcal{(C)}$ generated by $\rho$ and $M$.
 \section{Main results}
 
 In this section, we will count the numbers  of   orbits of $\langle\rho\rangle$ and $\langle\rho,M\rangle$ on $\mathcal{C}^\ast=\mathcal{C}\backslash\{0\}$, respectively. It is known that any  $\lambda$-constacyclic code is a direct sum of some irreducible  $\lambda$-constacyclic codes. Therefore, we first investigate the group actions on an irreducible  $\lambda$-constacyclic code. 
 Using the results  obtained, we proceeded to study the group actions on a general  $\lambda$-constacyclic code. 
 
 \subsection{ The actions of \texorpdfstring {$\langle \rho\rangle$}{}  on \texorpdfstring {$\mathcal{C}^\ast$}{}}
  By Lemma \ref{lem1}, the number of orbits of $\langle \rho\rangle$ on $\mathcal{C}^{\ast}$ is given by
 \begin{equation}\label{eq3}
 	 N_{\langle\rho\rangle}(\mathcal{C}^\ast) =\big| \langle \rho\rangle \backslash \mathcal{C}^{\ast} \big|=\frac{1}{tn}\sum_{r=1}^{tn} \big| {\rm Fix}(\rho^r) \big|,
 \end{equation}
  where ${\rm Fix}(\rho^r)=\{c \in \mathcal{C}^\ast | \rho^r(c)=c\}$.	
  
 \begin{lem}\label{lem2}
 	 Let $\mathcal{C}$ be an $[n,k]$ irreducible $\lambda$-constacyclic code over $\mathbb{F}_q$ with generating idempotent $\varepsilon_i$ and ${\rm ord}(\lambda)=t$.
 Suppose that   $\varepsilon_i$ corresponds to the $q$-cyclotomic coset $\{1+t\alpha_i, (1+t\alpha_i)q,\cdots,(1+t\alpha_i)q^{k-1}\}$. Then  $N_{\langle\rho\rangle}(\mathcal{C}^\ast)$ is equal to 
 $$\frac{(q^k-1)\gcd(1+t\alpha_i,n)}{tn}
 .$$
 Moreover, the number of nonzero weights of  $\mathcal{C}$ is less than or equal to  $N_{\langle\rho\rangle}(\mathcal{C}^\ast)$, with equality  if and only if for any two
 non-zero codewords $c_1,c_2\in\mathcal{C}$ with the same weight, there
 exists  an integer $i$ such that $\rho^i(c_1)=c_2$.
 \end{lem}
 \begin{proof}
 	 We only   calculate   the number of orbits of $\langle \rho\rangle$ on $\mathcal{C}^{\ast}$, and the remaining statements are from Proposition \ref{pro1}. From Eq. (\ref{eq3}), we intend to determine the values of $\big| {\rm Fix}(\rho^r) \big|$ for $1\leq r\leq tn$. Recall that  any codeword $c(x)$ in $\mathcal{C}^{\ast}$ can be expressed in the form  of 
 	 \begin{equation*}
 	 c(x)=	\sum_{j=0}^{k-1}(\sum_{h=0}^{k-1} c_h\zeta^{h(1+t\alpha_i)q^j})e_{(1+t\alpha_i)q^j}\in \mathcal{C}^{\ast},
 	 \end{equation*}
  where  $c_h\in\mathbb{F}_q$.  Since      $ e_{(1+t\alpha_i)q^j}=\frac{1}{n} \sum_{v=0}^{n-1}\zeta^{-v(1+t\alpha_i)q^j}x^v$, we get that 
  	 \begin{align*}
  	 	   \rho^r(e_{(1+t\alpha_i)q^j})&=\frac{1}{n} \sum_{v=0}^{n-1}\zeta^{-v(1+t\alpha_i)q^j}x^{v+r}\\
  	 	   &=\zeta^{r(1+t\alpha_i)q^j}\frac{1}{n} \sum_{v=0}^{n-1}\zeta^{-(v+r)(1+t\alpha_i)q^j}x^{v+r}\\
  	 	   &=\zeta^{r(1+t\alpha_i)q^j} e_{(1+t\alpha_i)q^j},
  	 \end{align*}
where the last equality holds because $\zeta^{-n(1+t\alpha_i)q^j}x^{n}=1$. Then we have 
	 \begin{align*}
\rho^r(	c(x))&=\rho^r\big(	\sum_{j=0}^{k-1}(\sum_{h=0}^{k-1} c_h\zeta^{h(1+t\alpha_i)q^j})e_{(1+t\alpha_i)q^j}\big)\\
&=	\sum_{j=0}^{k-1}(\sum_{h=0}^{k-1} c_h\zeta^{h(1+t\alpha_i)q^j}) \rho^r(e_{(1+t\alpha_i)q^j})\\
&=\sum_{j=0}^{k-1}\zeta^{r(1+t\alpha_i)q^j}(\sum_{h=0}^{k-1} c_h\zeta^{h(1+t\alpha_i)q^j}) e_{(1+t\alpha_i)q^j}.
\end{align*}
It follows that $\rho^r(c(x))=c(x)$ if and only if $\zeta^{r(1+t\alpha_i)q^j}=1$ for $0\leq j\leq k-1$. Since $\zeta$ is a primitive $tn$-th root of unity in $\mathbb{F}_{q^m}^{\ast}$ and $\gcd(tn,q)=1$, we get that  $\zeta^{r(1+t\alpha_i)q^j}=1$ if and only if $\frac{tn}{\gcd(1+t\alpha_i,n)}\mid r$. Therefore,
 \begin{eqnarray*}
\big| {\rm Fix}(\rho^r) \big| & = & \left\{\begin{array}{ll}
	q^k-1 , & \text { if }  \frac{tn}{\gcd(1+t\alpha_i,n)}\mid r, \\
	0, & \text { otherwise}.
	\end{array}\right.
\end{eqnarray*}	 
 By Eq.(\ref{eq3}), we get that the number of  orbits of $\langle \rho\rangle$ on $\mathcal{C}^{\ast}$ is equal  to 
  \begin{align*}
 	\frac{1}{tn}\sum_{r=1}^{tn} \big| {\rm Fix}(\rho^r) \big|
 	=\frac{1}{tn}\sum_{r=1, \atop \frac{tn}{\gcd(1+t\alpha_i,n)}\mid r}^{tn} q^k-1
 	=\frac{q^k-1}{tn}\sum_{r=1, \atop \frac{tn}{\gcd(1+t\alpha_i,n)}\mid r}^{tn} 1
 	=\frac{(q^k-1)\gcd(1+t\alpha_i,n)}{tn}.
 \end{align*} This completes the proof.
 \end{proof}
 \begin{exa}\label{ex1}
 	 Let $q=5, n=18, t=2$ in Lemma \ref{lem2}. All the 5-cyclotomic cosets are as follows:
 	 \begin{equation*}
 	 	 \{1,5,25,17,13,29\},~ \{3,15\},~ \{7,35,31,11,19,23\},~\{9\},~\{21,33\},~\{27\}.
 	 \end{equation*} Let $ \mathcal{I}'$ be the  irreducible $4$-constacyclic code over $\mathbb{F}_5$ corresponding to the 5-cyclotomic coset $ \{3,15\}$. Then by Lemma \ref{lem2}, the number of nonzero weights of $\mathcal{I}' $ is less than or equal to  $N_{\langle\rho\rangle}((\mathcal{I}')^\ast)$, where the latter is  equal to
 \begin{equation}
 	 \frac{(5^2-1)\gcd(3,18)}{2\times18}=2.
 \end{equation}  By using Magma \cite{bosma}, we  get that the weight distribution of
 $\mathcal{I}'$  is  $1+12x^{12}+12x^{18}$. This implies that the bound  given in Lemma \ref{lem2} is  tight.
 \end{exa}

We   proceed to investigate the group action of $\langle\rho\rangle$ on a general $\lambda$-constacyclic code of length $n$ over $\mathbb{F}_q$. 
Let $\gamma_1,\gamma_2,\cdots,\gamma_l$ be positive integers and let    $\beta_{\gamma_1},\beta_{\gamma_2},\cdots,\beta_{\gamma_l}$ be  integers with $0\leq \gamma_1<\gamma_2<\cdots\gamma_l\leq s$. Suppose that an irreducible $\lambda$-constacyclic code $\mathcal{I}_{\beta_{\gamma_i}}$ over $\mathbb{F}_q$ with generating idempotent $\varepsilon_{\beta_{\gamma_i}}$  corresponds to the  $q$-cyclotomic coset $\{1+t\alpha_{\beta_{\gamma_i}}, (1+t\alpha_{\beta_{\gamma_i}})q,\cdots,(1+t\alpha_{\beta_{\gamma_i}})q^{d_{\beta_{\gamma_i}}-1}\}$. We define
\begin{equation*}
	 \mathcal{C}_{\beta_{\gamma_1}\beta_{\gamma_2}\cdots\beta_{\gamma_l}}:= \mathcal{I}_{\beta_{\gamma_1}}\backslash \{0\} \bigoplus\mathcal{I}_{\beta_{\gamma_2}}\backslash \{0\} \bigoplus\cdots \bigoplus\mathcal{I}_{\beta_{\gamma_l}}\backslash \{0\}.
\end{equation*}Then $\langle \rho\rangle$ can act on $\mathcal{C}_{\beta_{\gamma_1}\beta_{\gamma_2}\cdots\beta_{\gamma_l}}$ in the same way as  the action on  $ \mathcal{C}^{\ast}$. 
 \begin{lem}\label{lem3}
 	  With the notations given above, then  $N_{\langle\rho\rangle}(\mathcal{C}_{\beta_{\gamma_1}\beta_{\gamma_2}\cdots\beta_{\gamma_l}})$ is equal  to 
 	  \begin{equation}
 	  	 \frac{\prod_{i=1}^{l}(q^{d_{\beta_{\gamma_i}}}-1)\gcd(1+t\alpha_{\beta_{\gamma_1}},\cdots,1+t\alpha_{\beta_{\gamma_l}},n)}{tn}.
 	  \end{equation}
 \end{lem}
\begin{proof}
	 Let $c=c_{\beta_{\gamma_1}}+c_{\beta_{\gamma_2}}+\cdots+c_{\beta_{\gamma_l}}\in  \mathcal{C}_{\beta_{\gamma_1}\beta_{\gamma_2}\cdots\beta_{\gamma_l}}$, where $c_{\beta_{\gamma_i}}\in \mathcal{I}_{\beta_{\gamma_i}}\backslash\{0\}$ for $i=1,\cdots,l$. From (\ref{main}), we suppose that 
	 \begin{equation*}
	 	 c_{\beta_{\gamma_i}}=\sum_{j=0}^{d_{\beta_{\gamma_i}}-1}(\sum_{h=0}^{d_{\beta_{\gamma_i}}-1} c_{h,\beta_{\gamma_i}}\zeta^{h(1+t\alpha_{\beta_{\gamma_i}})q^j})e_{(1+t\alpha_{\beta_{\gamma_i}})q^j} ~\text{for} ~i=1,\cdots,l.
	 \end{equation*}
 Then we have
 \begin{align*}
 	 \rho^r (c)&=\rho^r( c_{\beta_{\gamma_1}})+\cdots+\rho^r( c_{\beta_{\gamma_l}})\\
 	 &=\sum_{j_1=0}^{d_{\beta_{\gamma_1}}-1}\zeta^{r(1+t\alpha_{\beta_{\gamma_1}})q^{j_1}}(\sum_{h_1=0}^{d_{\beta_{\gamma_1}}-1} c_{h_1,\beta_{\gamma_1}}\zeta^{h_1(1+t\alpha_{\beta_{\gamma_1}})q^{j_1}})e_{(1+t\alpha_{\beta_{\gamma_1}})q^{j_1}}+\cdots\\
 	 &+\sum_{j_l=0}^{d_{\beta_{\gamma_l}}-1}\zeta^{r(1+t\alpha_{\beta_{\gamma_l}})q^{j_l}}(\sum_{h_l=0}^{d_{\beta_{\gamma_l}}-1} c_{h_l,\beta_{\gamma_l}}\zeta^{h_l(1+t\alpha_{\beta_{\gamma_l}})q^{j_l}})e_{(1+t\alpha_{\beta_{\gamma_l}})q^{j_l}}.
 \end{align*}
It follows that 
\begin{align*}
	 \rho^r(c) =c&\Leftrightarrow\rho^r(c_{\beta_{\gamma_1}})+\cdots+\rho^r(c_{\beta_{\gamma_l}})=c_{\beta_{\gamma_1}}+\cdots+c_{\beta_{\gamma_l}}\\
	 &\Leftrightarrow  \rho^r(c_{\beta_{\gamma_i}})=c_{\beta_{\gamma_i}} ~\text{for}~i=1,\cdots,l\\
	 &\Leftrightarrow\zeta^{r(1+t\alpha_{\beta_{\gamma_i}})q^{j_i}}=1 ~\text{for}~i=1,\cdots,l ~\text{and}~j_i=0,1,\cdots,d_{\beta_{\gamma_i}}-1\\
	 &\Leftrightarrow tn\mid r(1+t\alpha_{\beta_{\gamma_i}})q^{j_i} ~\text{for}~i=1,\cdots,l ~\text{and}~j_i=0,1,\cdots,d_{\beta_{\gamma_i}}-1\\
	  &\Leftrightarrow  tn\mid r(1+t\alpha_{\beta_{\gamma_i}}) ~\text{for}~i=1,\cdots,l \\
	   &\Leftrightarrow \frac{tn}{\gcd(1+t\alpha_{\beta_{\gamma_i}},n)} \mid r ~\text{for}~i=1,\cdots,l.
\end{align*}
Thereby, 
 \begin{eqnarray*}
	\big| {\rm Fix}(\rho^r) \big| & = & \left\{\begin{array}{ll}
	\prod_{i=1}^l	(q^{d_{\beta_{\gamma_i}}}-1), & \text { if } \frac{tn}{\gcd(1+t\alpha_{\beta_{\gamma_i}},n)} \mid r ~\text{for}~i=1,\cdots,l,\\
		0, & \text { otherwise}.
	\end{array}\right.
\end{eqnarray*}	 
By Eq.(\ref{eq3}), we have 
 \begin{align*}
N_{\langle\rho\rangle}(\mathcal{C}_{\beta_{\gamma_1}\beta_{\gamma_2}\cdots\beta_{\gamma_l}})
	&=\frac{1}{tn}\sum_{r=1}^{tn} \big| {\rm Fix}(\rho^r) \big|\\
	&=	\frac{\prod_{i=1}^l	(q^{d_{\beta_{\gamma_i}}}-1)}{tn} \sum_{r=1, \atop \frac{tn}{\gcd(1+t\alpha_{\beta_{\gamma_i}},n)} \mid r ~\text{for}~ i=1,\cdots,l}^{tn} 1\\
	&=\frac{\prod_{i=1}^l	(q^{d_{\beta_{\gamma_i}}}-1)}{tn} \sum_{r=1, \atop \frac{tn}{\gcd(1+t\alpha_{\beta_{\gamma_1}},\cdots,1+t\alpha_{\beta_{\gamma_l}},n)} \mid r }^{tn} 1\\
	&=  	 \frac{\prod_{i=1}^{l}(q^{d_{\beta_{\gamma_i}}}-1)\gcd(1+t\alpha_{\beta_{\gamma_1}},\cdots,1+t\alpha_{\beta_{\gamma_l}},n)}{tn}.
\end{align*} This completes the proof.
\end{proof}
\begin{thm}\label{thm1}
	 Let $\mathcal{C}$ be a $\lambda$-constacyclic code of length $n$ over $\mathbb{F}_q$  with  ${\rm ord}(\lambda)=t$.
	 Suppose that 
	 \begin{equation}
	 	 \mathcal{C}=\mathcal{I}_{\beta_1} \bigoplus \mathcal{I}_{\beta_2} \bigoplus  \cdots \bigoplus \mathcal{I}_{\beta_u},  
	 \end{equation}
 where $0\leq \beta_1< \beta_2 <\cdots< \beta_u\leq s$, and    $\mathcal{I}_{\beta_i}$   corresponds to the $q$-cyclotomic coset $\{1+t\alpha_{\beta_i}, (1+t\alpha_{\beta_i})q,\cdots,(1+t\alpha_{\beta_i})q^{d_{\beta_i}-1}\}$.
  Then   $N_{\langle\rho\rangle}(\mathcal{C}^\ast)$ is equal  to 
 \begin{equation*}
 	 \sum_{\{\gamma_1,\cdots,\gamma_l\}\in\{1,\cdots,u\}\atop 1\leq \gamma_1<\gamma_2<\cdots<\gamma_l\leq u}   \frac{\prod_{i=1}^{l}(q^{d_{\beta_{\gamma_i}}}-1)\gcd(1+t\alpha_{\beta_{\gamma_1}},\cdots,1+t\alpha_{\beta_{\gamma_l}},n)}{tn}.
 \end{equation*}
 Moreover, the number of nonzero weights of  $\mathcal{C}$ is less than or equal to $N_{\langle\rho\rangle}(\mathcal{C}^\ast)$, with equality  if and only if for any two
non-zero codewords $c_1,c_2\in\mathcal{C}$ with the same weight, there
exists  an integer $i$ such that $\rho^i(c_1)=c_2$.
\end{thm}
\begin{proof}
	  We only  calculate the number of orbits of $\langle \rho\rangle$ on $\mathcal{C}^\ast$, and the remaining statements are from Proposition \ref{pro1}. By a naturally way, we know that $\mathcal{C}^\ast$  can be decomposed into
	  \begin{equation*}
	  		\bigcup_{\{\gamma_1,\cdots,\gamma_l\}\in\{1,\cdots,u\}\atop 1\leq \gamma_1<\gamma_2<\cdots<\gamma_l\leq u}  \mathcal{C}_{\beta_{\gamma_1}\beta_{\gamma_2}\cdots\beta_{\gamma_l}}.
	  \end{equation*}
	  From Lemma \ref{lem3}, the desired results are obtained.
\end{proof}
\begin{coro}\label{co1}
	 Let $\mathcal{C}=\mathcal{I}_{\beta_1}\bigoplus\mathcal{I}_{\beta_2}$ be a $\lambda$-constacyclic code over $\mathbb{F}_q$,where $0\leq \beta_1< \beta_2 \leq s$, and   $\mathcal{I}_{\beta_i}$ corresponds to the $q$-cyclotomic coset $\{1+t\alpha_{\beta_i},(1+t\alpha_{\beta_i})q,\cdots,(1+t\alpha_{\beta_i})q^{d_{\beta_i}-1}\}$ for $1\leq i\leq 2$. Then  $N_{\langle\rho\rangle}(\mathcal{C}^\ast)$  is equal  to 
	 \begin{equation*}
	  \frac{\prod\limits_{i=1}^{2}(q^{d_{\beta_i}}-1)\gcd(1+t\alpha_{\beta_1},1+t\alpha_{\beta_2},n)}{tn}+ \frac{(q^{d_{\beta_1}}-1)\gcd(1+t\alpha_{\beta_1},n)}{tn}+\frac{(q^{d_{\beta_2}}-1)\gcd(1+t\alpha_{\beta_2},n)}{tn}.
	 \end{equation*} Moreover, the number of nonzero weights of  $\mathcal{C}$ is less than or equal to $N_{\langle\rho\rangle}(\mathcal{C}^\ast)$, with equality  if and only if for any two
 non-zero codewords $c_1,c_2\in\mathcal{C}$ with the same weight, there
exists  an integer $i$ such that $\rho^i(c_1)=c_2$.
\end{coro}
\begin{proof}
	 The desired results are directly obtained by Theorem \ref{thm1}.
\end{proof}

\begin{exa}\label{ex2}
	 Let $q=3,~n=65,~t=2$ in Corollary \ref{co1}. Let $\mathcal{I}'$ and $\mathcal{I}''$ be the irreducible $2$-constacyclic codes over $\mathbb{F}_3$ corresponding to the 3-cyclotomic cosets $\{65\}$ and $\{5,15,45\}$, respectively. By Corollary \ref{co1}, 
	the number of nonzero weights of   $\mathcal{I}'\bigoplus\mathcal{I}''$ is less than or equal to  $N_{\langle\rho\rangle}((\mathcal{I}'\bigoplus\mathcal{I}'')^\ast)$, where  the latter is  equal to
	 \begin{equation*}
	\frac{(3^3-1)(3-1)\gcd(65,65,5)}{2\times65}+\frac{(3^3-1)\gcd(65,5)}{2\times65}+\frac{(3-1)\gcd(65,65)}{2\times65}=4
	 \end{equation*} 
 By using Magma \cite{bosma}, we  get that the weight distribution of
$\mathcal{I}'\bigoplus\mathcal{I}''$ is  $1+26x^{35}+26x^{45}+26x^{50}+2x^{65}$. This implies that the bound  given in Corollary \ref{co1} ( or Theorem \ref{thm1} ) is  tight.
\end{exa}
\subsection{ The actions of \texorpdfstring {$\langle \rho,M\rangle$}{}  on \texorpdfstring {$\mathcal{C}^\ast$}{}}
In this subsection, we discuss the group action of $\langle \rho,M\rangle$ on $\mathcal{C}^\ast$, where
$\langle \rho,M\rangle$ is a subgroup of  ${\rm Aut}\mathcal{(C)}$ generated by $\rho$ and $M=\{\sigma_b|b\in\mathbb{F}_q^{\ast}\}$.  As is easily  checked, $\rho\sigma_b=\sigma_b\rho$ for any $b\in \mathbb{F}_q^{\ast}$.
\begin{lem}\cite{chen1}
	 The subgroup $\langle \rho,M\rangle$   of  ${\rm Aut}\mathcal{(C)}$ is the direst product of $ \langle \rho\rangle$ and $M$, in symbols,
	 $$ \langle \rho,M\rangle=\langle \rho\rangle \times M .$$
	 In particular, $\big|\langle \rho,M\rangle \big|=tn(q-1)$.
\end{lem}

\begin{lem}\label{lem5}
	 Let $\mathcal{C}$ be an $[n,k]$ irreducible $\lambda$-constacyclic code over $\mathbb{F}_q$ with generating idempotent $\varepsilon_i$ and ${\rm ord}(\lambda)=t$.
	Suppose that $\varepsilon_i$ corresponds to the $q$-cyclotomic coset $\{1+t\alpha_i, (1+t\alpha_i)q,\cdots,(1+t\alpha_i)q^{d_i-1}\}$. Then   $N_{\langle\rho,M\rangle}(\mathcal{C}^\ast)$ is equal to 
	$$\frac{(q^k-1)\gcd((q-1)(1+t\alpha_i),tn)}{(q-1)tn}.$$
	Moreover, the number of nonzero weights of  $\mathcal{C}$ is less than or equal to   $N_{\langle\rho,M\rangle}(\mathcal{C}^\ast)$, with equality  if and only if for any two
	non-zero codewords $c_1,c_2\in\mathcal{C}$ with the same weight, there
	exist an  integer $i$ and an element $b\in \mathbb{F}_q^{\ast}$ such that $\rho^i(bc_1)=c_2$. 
\end{lem}
\begin{proof}
	  As is easily  observed, $M$ is isomorphic to the cyclic group $\mathbb{F}_q^{\ast}$. Let $\xi$ be a primitive element of $\mathbb{F}_q^{\ast}$. Then $\sigma_\xi$ is a generator of $M$, i.e., $M=\{\sigma_{\xi}^r|1\leq r\leq q-1\}$. Recall that $\langle \rho \rangle \backslash \mathcal{C}^{\ast}=\{ \langle \rho \rangle(c)|c\in\mathcal{C}^{\ast}\}$ denotes the set of orbits of $\langle \rho \rangle$ on $\mathcal{C}^{\ast}$, where $  \langle \rho \rangle(c)=\{\rho^r(c)|1\leq r\leq tn\}$. Then we naturally give rise to the  action of $M$ on  $\langle \rho \rangle \backslash \mathcal{C}^{\ast}$ as follows:
	   \begin{align*}
	 & M\times  \langle \rho \rangle \backslash \mathcal{C}^{\ast}\longrightarrow  \langle \rho \rangle \backslash \mathcal{C}^{\ast}\\
	 & (\sigma_\xi^r,\langle\rho\rangle(c)) \longrightarrow  \langle \rho \rangle (\xi^rc).
	   \end{align*}
   By Lemma \ref{lem4}, $N_{\langle\rho,M\rangle}(\mathcal{C}^\ast)=N_{M}(\mathcal{\langle \rho \rangle \backslash \mathcal{C}^{\ast}})$.
   Therefore, by Lemma \ref{lem1},
   \begin{align}\label{eq6}
   	N_{\langle\rho,M\rangle}(\mathcal{C}^\ast) = \frac{1}{q-1}\sum_{r=1}^{q-1}\big| {\rm Fix}(\sigma_{\xi}^{r}) \big|,
   \end{align}where ${\rm Fix} (\sigma_{\xi}^{r}) =\{  \langle \rho \rangle(c)\in  \langle \rho \rangle \backslash\mathcal{C}^{\ast}|  \langle \rho \rangle(\xi^rc)= \langle \rho \rangle(c)  \}$.  Next, we intend to calculate  the value of  $\big|{\rm Fix} (\sigma_{\xi}^{r})\big|$ for $1\leq r\leq q-1$.    Recall that  any codeword $c(x)$ in $\mathcal{C}^{\ast}$ has the  form of
\begin{equation}\label{eq5}
c(x)=	\sum_{j=0}^{k-1}(\sum_{h=0}^{k-1} c_h\zeta^{h(1+t\alpha_i)q^j})e_{(1+t\alpha_i)q^j}\in \mathcal{C}^{\ast},
\end{equation}
where  $c_h\in\mathbb{F}_q$. The condition $ \langle \rho \rangle(\xi^rc)= \langle \rho \rangle(c) $  is equivalent to requiring that there exists an integer $z$ such that $\rho^z(c)=\xi^rc$. By Eq. (\ref{eq5}) and the proof of Lemma \ref{lem2}, we have
\begin{equation*}
	\rho^z c(x)=	\sum_{j=0}^{k-1}\zeta^{ z(1+t\alpha_i)q^j} (\sum_{h=0}^{k-1} c_h\zeta^{h(1+t\alpha_i)q^j})e_{(1+t\alpha_i)q^j} ~\text{ and }~
	  \xi^r c(x)=	\sum_{j=0}^{k-1}\xi^r (\sum_{h=0}^{k-1} c_h\zeta^{h(1+t\alpha_i)q^j})e_{(1+t\alpha_i)q^j} 
\end{equation*}
Thereby, $	\rho^z c(x)= 	  \xi^r c(x)$  if and only if $\zeta^{ z(1+t\alpha_i)q^j}=\xi^r$ for $0\leq j\leq k-1$. Since $\xi\in \mathbb{F}_q^{\ast}$, $\zeta^{ z(1+t\alpha_i)q^j}=\xi^r$ for $0\leq j\leq k-1$ if and only if  $\zeta^{ z(1+t\alpha_i)}=\xi^r$.

Let $\omega$ be a primitive element  of $\mathbb{F}_{q^m}^{\ast}$. Recall that $\zeta$ is a primitive $tn$-th root of unity in $\mathbb{F}_{q^m}^{\ast}$ and $\xi$ is a primitive element of $\mathbb{F}_{q}^{\ast}$. Set $\zeta=\omega^{\frac{q^m-1}{tn}}$ and $\xi=\omega^{\frac{q^m-1}{q-1}}$. Then we have 
\begin{align*}
	  \zeta^{ z(1+t\alpha_i)}=\xi^r &\Leftrightarrow  \omega^{\frac{q^m-1}{tn}(1+t\alpha_i)z} =\omega^{\frac{q^m-1}{q-1}r}\\
	  &
	  \Leftrightarrow  \big\langle \omega^{\frac{q^m-1}{q-1}r} \big\rangle  \subseteq \big\langle \omega^{\frac{q^m-1}{tn}(1+t\alpha_i)}\big\rangle \\
	  &\Leftrightarrow {\rm ord}( \omega^{\frac{q^m-1}{q-1}r}  )\mid  {\rm ord} (\omega^{\frac{q^m-1}{tn}(1+t\alpha_i)})\\
	  &\Leftrightarrow \gcd\big(q^m-1, \frac{q^m-1}{tn}(1+t\alpha_i)\big ) \mid \gcd (q^m-1,\frac{q^m-1}{q-1}r)  \\
	  &\Leftrightarrow \frac{q^m-1}{tn}\gcd (1+t\alpha_i,n)\mid  \frac{(q^m-1)}{q-1}\gcd(r,q-1)\\
	  &\Leftrightarrow \frac{q-1}{\gcd(r,q-1)}\mid \frac{tn}{\gcd (1+t\alpha_i,n)},
\end{align*}
where  $\big\langle\omega^{\frac{q^m-1}{q-1}r} \big\rangle$ and $ \big\langle \omega^{\frac{q^m-1}{tn}(1+t\alpha_i)}\big\rangle$ denote the cyclic subgroups  of $\mathbb{F}_{q^m}^{\ast}$ generated by  $\omega^{\frac{q^m-1}{q-1}r}$ and $\omega^{\frac{q^m-1}{tn}(1+t\alpha_i)} $, respectively. It follows that  $ \langle \rho \rangle(\xi^rc)= \langle \rho \rangle(c) $ if and only if $\frac{q-1}{\gcd(r,q-1)}\mid \frac{tn}{\gcd (1+t\alpha_i,n)}$. By Lemma \ref{lem2}, $N_{\langle\rho\rangle}(\mathcal{C}^\ast)=\big|\langle \rho\rangle\backslash \mathcal{C}^{\ast}\big|=\frac{(q^k-1)\gcd(1+t\alpha_i,n)}{tn}$, then we have 
 \begin{eqnarray*}
	\big| {\rm Fix}(\sigma_{\xi}^r) \big| & = & \left\{\begin{array}{ll}
	 \frac{(q^k-1)\gcd(1+t\alpha_i,n)}{tn}, & \text { if } \frac{q-1}{\gcd(r,q-1)}\mid \frac{tn}{\gcd (1+t\alpha_i,n)}, \\
		0, & \text { otherwise}.
	\end{array}\right.
\end{eqnarray*}	By Eq. (\ref{eq6}), we get that 
\begin{align*}
	 N_{\langle\rho,M\rangle}(\mathcal{C}^\ast)   &=\frac{1}{q-1}\sum_{r=1}^{q-1}\big| {\rm Fix}(\sigma_{\xi}^{r}) \big| \\
	    	 &=\frac{(q^k-1)\gcd(1+t\alpha_i,n)}{(q-1)tn} \sum_{r=1,\atop \frac{q-1}{\gcd(r,q-1)}| \frac{tn}{\gcd (1+t\alpha_i,n)}}^{q-1}1\\
	    	 	 &=\frac{(q^k-1)\gcd(1+t\alpha_i,n)}{(q-1)tn} \sum_{\frac{q-1}{g}\mid q-1,\atop \frac{q-1}{g}| \frac{tn}{\gcd (1+t\alpha_i,n)}}\varphi(\frac{q-1}{g})\\ 
	    	 	 &=\frac{(q^k-1)\gcd(1+t\alpha_i,n)}{(q-1)tn}\gcd\big(q-1,\frac{tn}{\gcd (1+t\alpha_i,n)} \big)\\
	    	 	 &= \frac{(q^k-1)\gcd\big((q-1)(1+t\alpha_i),tn\big)}{(q-1)tn}.
\end{align*}Note that for $g\mid q-1$, the number of $r $ satisfying  $1\leq r \leq q-1$ and $\gcd(r,q-1)=g$ is equal to $\varphi(\frac{q-1}{g})$, where $\varphi$ is the Euler Phi function.    This completes the proof.
\end{proof}

 \begin{exa}\label{ex3}
	Let $q=7, n=32, t=3$ in Lemma \ref{lem5}. Let $\mathcal{I}'$  be the irreducible $2$-constacyclic code over $\mathbb{F}_7$ corresponding to the 7-cyclotomic coset $\{10,70\}$. Then 
	the number of nonzero weights of   $\mathcal{I}'$ is less than or equal to    $N_{\langle\rho,M\rangle}((\mathcal{I}')^\ast)$, where  the latter is  equal to
	\begin{equation*}
	\frac{(49-1)\gcd(6\times10,96)}{6\times3\times32}=1
	\end{equation*} Hence,  $\mathcal{I}'$ must be an  one weight code. 
	By using Magma \cite{bosma}, we  get  that the weight distribution of
	$\mathcal{I}'$ is  $1+48x^{28}$. This implies that the bound  given in Lemma \ref{lem5}  is  tight.
\end{exa}

We   proceed to investigate the group action of $\langle\rho,M\rangle$ on a general $\lambda$-constacyclic code of length $n$ over $\mathbb{F}_q$.
For convenience,  we recall some  defined notions.
Let $\gamma_1,\gamma_2,\cdots,\gamma_l$ be positive integers and let    $\beta_{\gamma_1},\beta_{\gamma_2},\cdots,\beta_{\gamma_l}$ be  integers with $0\leq \gamma_1<\gamma_2<\cdots,\gamma_l\leq s$. Suppose an irreducible $\lambda$-constacyclic code $\mathcal{I}_{\beta_{\gamma_i}}$ over $\mathbb{F}_q$ with generating idempotent $\varepsilon_{\beta_{\gamma_i}}$  corresponds to the $q$-cyclotomic coset $\{1+t\alpha_{\beta_{\gamma_i}}, (1+t\alpha_{\beta_{\gamma_i}})q,\cdots,(1+t\alpha_{\beta_{\gamma_i}})q^{d_{\beta_{\gamma_i}}-1}\}$. We define
\begin{equation*}
		\mathcal{C}_{\beta_{\gamma_1}\beta_{\gamma_2}\cdots\beta_{\gamma_l}}:= \mathcal{I}_{\beta_{\gamma_1}}\backslash \{0\} \bigoplus\mathcal{I}_{\beta_{\gamma_2}}\backslash \{0\} \bigoplus\cdots \bigoplus\mathcal{I}_{\beta_{\gamma_l}}\backslash \{0\}.
\end{equation*}Then we can get the number of orbits of $\langle\rho,M\rangle$ on $		\mathcal{C}_{\beta_{\gamma_1}\beta_{\gamma_2}\cdots\beta_{\gamma_l}}$. 
\begin{lem}\label{lem7}
	 With the notations given above,  $N_{\langle\rho,M\rangle}(\mathcal{C}_{\beta_{\gamma_1}\beta_{\gamma_2}\cdots\beta_{\gamma_l}} )$  is equal to 
	 \begin{equation*}
	 	 \frac{\prod_{i=1}^{l}(q^{d_{\beta{\gamma_i}}}-1)}{(q-1)tn}\gcd(1+t\alpha_{\beta{\gamma_1}},\cdots,1+t\alpha_{\beta{\gamma_l}},n)\gcd(q-1,\frac{tn}{\gcd(1+t\alpha_{\beta{\gamma_1}},n)},\cdots,\frac{tn}{\gcd(1+t\alpha_{\beta{\gamma_l}},n)}).
	 \end{equation*}
\end{lem}
\begin{proof}
By Eq. (\ref{eq6}),   we get that 
 \begin{align*}
	 N_{\langle\rho,M\rangle}(\mathcal{C}_{\beta_{\gamma_1}\beta_{\gamma_2}\cdots\beta_{\gamma_l}} ) =\frac{1}{q-1}\sum_{r=1}^{q-1}\big| {\rm Fix}(\sigma_{\xi}^{r}) \big|,
\end{align*} where ${\rm Fix}   (\sigma_\xi^r)=\{\langle \rho\rangle(c)\in \langle \rho \rangle \backslash\mathcal{C}_{\beta_{\gamma_1}\beta_{\gamma_2}\cdots\beta_{\gamma_l}}|\langle\rho\rangle(c)=\langle\rho\rangle(\xi^rc)\}$.
We only  calculate the values of  $\big|{\rm Fix}   (\sigma_\xi^r)\big|$ for $1\leq r \leq q-1$. By the proof of Lemma \ref{lem5}, the condition $ \langle \rho \rangle(\xi^rc)= \langle \rho \rangle(c) $  is equivalent to requiring that there exists an integer $z$ such that $\rho^z(c)=\xi^rc$.

 Let $c=c_{\beta_{\gamma_1}}+c_{\beta_{\gamma_2}}+\cdots+c_{\beta_{\gamma_l}}\in  \mathcal{C}_{\beta_{\gamma_1}\beta_{\gamma_2}\cdots\beta_{\gamma_l}} $, where $c_{\beta_{\gamma_i}}\in \mathcal{I}_{\beta_{\gamma_i}}\backslash\{0\}$ for $i=1,\cdots,l$.
Then  $ \rho ^z(c)= \xi^rc $ if and only if 
\begin{eqnarray}\label{eq7}
\rho ^z   (c_{\beta_{\gamma_i}})= \xi^rc_{\beta_{\gamma_i}}
\end{eqnarray} for $i=1,\cdots,l$. Again by the proof of Lemma \ref{lem5},  we get that  the equalities (\ref{eq7}) hold if and only if 
\begin{equation}\label{eq8}
\frac{q-1}{\gcd(r,q-1)}\mid \frac{tn}{\gcd (1+t\alpha_{\beta_{\gamma_i}},n)}	 
\end{equation}
 for  $i=1,\cdots,l$. Thus, if equalities (\ref{eq8}) hold, then
	  \begin{equation}
\big| {\rm Fix}(\sigma_{\xi}^{r}) \big|=	\frac{\prod_{i=1}^{l}(q^{d_{\beta_{\gamma_i}}}-1)\gcd(1+t\alpha_{\beta_{\gamma_1}},\cdots,1+t\alpha_{\beta_{\gamma_l}},n)}{tn}
\end{equation} by Lemma \ref{lem3}; otherwise, $ \big| {\rm Fix}(\sigma_{\xi}^{r}) \big|=0$. 
Thereby,
\begin{align*}
 N_{\langle\rho,M\rangle}(\mathcal{C}_{\beta_{\gamma_1}\beta_{\gamma_2}\cdots\beta_{\gamma_l}} ) &= \frac{1}{q-1}\sum_{r=1}^{q-1}\big| {\rm Fix}(\sigma_{\xi}^{r}) \big|\\
	&= \frac{\prod_{i=1}^{l}(q^{d_{\beta_{\gamma_i}}}-1)\gcd(1+t\alpha_{\beta_{\gamma_1}},\cdots,1+t\alpha_{\beta_{\gamma_l}},n)}{(q-1)tn} \sum_{r=1,\frac{q-1}{\gcd(r,q-1)}\mid \frac{tn}{\gcd (1+t\alpha_{\beta_{\gamma_i}},n)}\atop i=1,\cdots,l}^{q-1}1\\
	&=\frac{\prod_{i=1}^{l}(q^{d_{\beta_{\gamma_i}}}-1)\gcd(1+t\alpha_{\beta_{\gamma_1}},\cdots,1+t\alpha_{\beta_{\gamma_l}},n)}{(q-1)tn}  \sum_{\frac{q-1}{g}\mid q-1\atop\frac{q-1}{g}\mid \frac{tn}{\gcd (1+t\alpha_{\beta_{\gamma_i}},n),}i=1,\cdots,l  }\varphi(\frac{q-1}{g})\\
	&=\frac{\prod_{i=1}^{l}(q^{d_{\beta_{\gamma_i}}}-1)\gcd(1+t\alpha_{\beta_{\gamma_1}},\cdots,1+t\alpha_{\beta_{\gamma_l}},n)}{(q-1)tn}\\
	&\cdot\gcd(q-1,\frac{tn}{\gcd(1+t\alpha_{\beta{\gamma_1}},n)},
	\cdots,\frac{tn}{\gcd(1+t\alpha_{\beta{\gamma_l}},n)}).
\end{align*}
\end{proof} 

\begin{thm}\label{thm2}
 Let $\mathcal{C}$ be a $\lambda$-constacyclic code of length $n$ over $\mathbb{F}_q$  with  ${\rm ord}(\lambda)=t$.
Suppose that 
\begin{equation}
	\mathcal{C}=\mathcal{I}_{\beta_1} \bigoplus \mathcal{I}_{\beta_2} \bigoplus  \cdots \bigoplus \mathcal{I}_{\beta_u},  
\end{equation}
where $0\leq \beta_1< \beta_2 <\cdots< \beta_u\leq s$, and   $\mathcal{I}_{\beta_i}$  corresponds to the $q$-cyclotomic coset $\{1+t\alpha_{\beta_i}, (1+t\alpha_{\beta_i})q,\cdots,(1+t\alpha_{\beta_i})q^{d_{\beta_i}-1}\}$. Then   $N_{\langle\rho,M\rangle}(\mathcal{C}^\ast)$   is equal  to 
	\begin{align*}
	\sum_{\{\gamma_1,\cdots,\gamma_l\}\in\{1,\cdots,u\}\atop 1\leq \gamma_1<\gamma_2<\cdots<\gamma_l\leq u}   &\frac{\prod_{i=1}^{l}(q^{d_{\beta{\gamma_i}}}-1)}{(q-1)tn}\gcd(n,1+t\alpha_{\beta{\gamma_1}},\cdots,1+t\alpha_{\beta{\gamma_l}})\\
		&\cdot\gcd(q-1,\frac{tn}{\gcd(n,1+t\alpha_{\beta{\gamma_1}})},
		\cdots,\frac{tn}{\gcd(n,1+t\alpha_{\beta{\gamma_l}})}). 
	\end{align*}
	Moreover, the number of nonzero weights of  $\mathcal{C}$ is less than or equal to  $N_{\langle\rho,M\rangle}(\mathcal{C}^\ast)$, with equality  if and only if for any two
	non-zero codewords $c_1,c_2\in\mathcal{C}$ with the same weight, there
	exist an  integer $i$ and an element $b\in \mathbb{F}_q^{\ast}$ such that $\rho^i(bc_1)=c_2$.
\end{thm}
\begin{proof}
	 Similar to the proof of Theorem \ref{thm1}, by Proposition \ref{pro1} and Lemma \ref{lem7},  the desired results are obtained.
\end{proof}

\begin{coro}\label{co2}
	Let $\mathcal{C}=\mathcal{I}_{\beta_1}\bigoplus\mathcal{I}_{\beta_2}$ be a $\lambda$-constacyclic code over $\mathbb{F}_q$,where $0\leq \beta_1< \beta_2 \leq s$, and   $\mathcal{I}_{\beta_i}$ corresponds to the $q$-cyclotomic coset $\{1+t\alpha_{\beta_i},(1+t\alpha_{\beta_i})q,\cdots,(1+t\alpha_{\beta_i})q^{d_{\beta_i}-1}\}$ for $1\leq i\leq 2$. Then  $N_{\langle\rho,M\rangle}(\mathcal{C}^\ast)$  is equal  to 
	\begin{align*}
		&\frac{\prod\limits_{i=1}^{2}(q^{d_{\beta_i}}-1)\gcd(n,1+t\alpha_{\beta_1},1+t\alpha_{\beta_2})}{(q-1)tn}  \gcd(   q-1,\frac{tn}{\gcd(n,1+t\alpha_{\beta_1})}, \frac{tn}{\gcd(n,1+t\alpha_{\beta_2})}  ) \\
		& +\frac{(q^{d_{\beta_1}}-1)\gcd((q-1)(1+t\alpha_{\beta_1}),tn)}{(q-1)tn} +\frac{(q^{d_{\beta_2}}-1)\gcd((q-1)(1+t\alpha_{\beta_2}),tn)}{(q-1)tn}.
	\end{align*}
	Moreover, the number of nonzero weights of  $\mathcal{C}$ is less than or equal to $N_{\langle\rho,M\rangle}(\mathcal{C}^\ast)$, with equality  if and only if for any two
non-zero codewords $c_1,c_2\in\mathcal{C}$ with the same weight, there
exist an  integer $i$ and an element $b\in \mathbb{F}_q^{\ast}$ such that $\rho^i(bc_1)=c_2$.
\end{coro}
\begin{proof}
	The desired results are directly obtained by Theorem \ref{thm2}.
\end{proof}
\begin{exa}\label{ex4}
	Let $q=3,~n=91,~t=2$ in Corollary. Let $\mathcal{I}'$ and $\mathcal{I}''$ be the irreducible $2$-constacyclic codes over $\mathbb{F}_3$ corresponding to the 3-cyclotomic cosets	 $\{91\}$ and $\{7,21,63\}$, respectively. By Corollary \ref{co2}, 
	the number of nonzero weights of   $\mathcal{I}'\bigoplus\mathcal{I}''$ is less than or equal to  $N_{\langle\rho,M\rangle}((\mathcal{I}'\bigoplus\mathcal{I}'')^\ast)$, where the latter is  equal to
	\begin{align*}
		&\frac{(3^3-1)(3-1)\gcd(91,7,91)}{2\times2\times91}\gcd(2,2,26)+\\
		&\frac{(3^3-1)\gcd(2\times7,2\times91)}{2\times2\times91}+\frac{(3-1)\gcd(2\times91,2\times91)}{2\times2\times91}=4
	\end{align*} 
	By using Magma \cite{bosma}, we  get that the weight distribution of
	$\mathcal{I}'\bigoplus\mathcal{I}''$ is  $1+26x^{49}+26x^{63}+26x^{70}+2x^{91}$. This implies that the bound  given in Corollary \ref{co2} ( or Theorem \ref{thm2} ) is  tight.
\end{exa}

Comparing the expressions for $N_{\langle\rho\rangle}(\mathcal{C}^\ast)$ in Theorem \ref{thm1} and $N_{\langle\rho,M\rangle}(\mathcal{C}^\ast)$ in Theorem \ref{thm2}, they  differ only by a factor of $$\Delta_{\beta_{\gamma_1}\beta_{\gamma_2}\cdots\beta_{\gamma_l}}:=\frac{\gcd(q-1,\frac{tn}{\gcd(n,1+t\alpha_{\beta{\gamma_1}})},\cdots,\frac{tn}{\gcd(n,1+t\alpha_{\beta{\gamma_l}})})}{q-1}$$ for each term in the summation. It obvious that $\Delta_{\beta_{\gamma_1}\beta_{\gamma_2}\cdots\beta_{\gamma_l}}\leq 1$. Thereby, the bound in Theorem \ref{thm2} is sharper than the bound in Theorem \ref{thm1} in general, but it is also more complex.
 Moreover,  $N_{\langle\rho\rangle}(\mathcal{C}^\ast)$= $N_{\langle\rho,M\rangle}(\mathcal{C}^\ast)$ if $	\Delta_{\beta_{\gamma_1}\beta_{\gamma_2}\cdots\beta_{\gamma_l}}=1$ for all  $\{\gamma_1,\cdots,\gamma_l\}\in\{1,\cdots,u\}$ and $ 1\leq \gamma_1<\gamma_2<\cdots<\gamma_l\leq u$. The following proposition give a sufficient condition  to guarantee   $N_{\langle\rho\rangle}(\mathcal{C}^\ast)$= $N_{\langle\rho,M\rangle}(\mathcal{C}^\ast)$.
\begin{prop}\label{pro2}
	 	Let $\mathcal{C}$ be a $\lambda$-constacyclic code of length $n$ over $\mathbb{F}_q$.
	 Suppose that 
	 \begin{equation}
	 	\mathcal{C}=\mathcal{I}_{\beta_1} \bigoplus \mathcal{I}_{\beta_2} \bigoplus  \cdots \bigoplus \mathcal{I}_{\beta_u},  
	 \end{equation}
	 where $0\leq \beta_1< \beta_2 <\cdots< \beta_u\leq s$, and    $\mathcal{I}_{\beta_i}$ corresponds to the $q$-cyclotomic coset $\{1+t\alpha_{\beta_i},(1+t\alpha_{\beta_i})q,\cdots,(1+t\alpha_{\beta_i})q^{d_{\beta_i}-1}\}$ for $0\leq i\leq u$.  Then   $N_{\langle\rho\rangle}(\mathcal{C}^\ast)$= $N_{\langle\rho,M\rangle}(\mathcal{C}^\ast)$  if  		 
	 \begin{equation}\label{suff}
	\frac{q-1}{t} {\rm lcm}\big(\gcd(n,1+t\alpha_{\beta{\gamma_1}}),\cdots,\gcd(n,1+t\alpha_{\beta{\gamma_l}})     \big) \big| n
	 \end{equation} for  $\{\gamma_1,\cdots,\gamma_l\}\in\{1,\cdots,u\}$ and $ 1\leq \gamma_1<\gamma_2<\cdots<\gamma_l\leq u$. In particular, if $t=q-1$, then  $N_{\langle\rho\rangle}(\mathcal{C}^\ast)$= $N_{\langle\rho,M\rangle}(\mathcal{C}^\ast)$.
\end{prop}
\begin{proof}For  $\{\gamma_1,\cdots,\gamma_l\}\in\{1,\cdots,u\}$ and $ 1\leq \gamma_1<\gamma_2<\cdots<\gamma_l\leq u$, we have
\begin{align*} 
	 	\Delta_{\beta_{\gamma_1}\beta_{\gamma_2}\cdots\beta_{\gamma_l}}=1 &\Leftrightarrow q-1 \big| \gcd(\frac{tn}{\gcd(n,1+t\alpha_{\beta{\gamma_1}})},
	 		\cdots,\frac{tn}{\gcd(n,1+t\alpha_{\beta{\gamma_l}})})\\
	 		&\Leftrightarrow q-1  \big| \frac{tn}{{\rm lcm}\big(\gcd(n,1+t\alpha_{\beta{\gamma_1}}),\cdots,\gcd(n,1+t\alpha_{\beta{\gamma_l}})     \big)}\\
	 		&\Leftrightarrow\frac{q-1}{t} {\rm lcm}\big(\gcd(n,1+t\alpha_{\beta{\gamma_1}}),\cdots,\gcd(n,1+t\alpha_{\beta{\gamma_l}})     \big) \big| n.
\end{align*}In particular, if $t=q-1$, then  $\Delta_{\beta_{\gamma_1}\beta_{\gamma_2}\cdots\beta_{\gamma_l}}=1$ for all $\{\gamma_1,\cdots,\gamma_l\}\in\{1,\cdots,u\}$ and $ 1\leq \gamma_1<\gamma_2<\cdots<\gamma_l\leq u$ and $N_{\langle\rho\rangle}(\mathcal{C}^\ast)$= $N_{\langle\rho,M\rangle}(\mathcal{C}^\ast)$.  This completes the proof.
\end{proof}
\begin{exa}
	 As is easily checked, the values of $t$ and $1+t\alpha_i$ in  Examples \ref{ex1} to \ref{ex4}  satisfy the sufficient condition (\ref{suff}) in Proposition \ref{pro2}.  Then $N_{\langle\rho\rangle}(\mathcal{C}^\ast)$= $N_{\langle\rho,M\rangle}(\mathcal{C}^\ast)$.
\end{exa}
\section{Conclusion}
In this paper, we calculate the numbers of orbits of $\langle\rho\rangle$ and $\langle\rho,M\rangle$
 on a simple-root $\lambda$-constacyclic code over $\mathbb{F}_q$, respectively.
 By the virtue of Proposition \ref{pro1}, we get   explicit bounds  on the number of nonzero weights of a simple-root $\lambda$-constacyclic code in Theorems \ref{thm1} and  \ref{thm2}. The bound in Theorem \ref{thm2} is sharper than that of   Theorem  \ref{thm1} in general, but it is also more complex.  We present a sufficient condition to  guarantee the equality of these two bounds.  Some irreducible and reducible $\lambda$-constacyclic codes meeting this bound are presented, revealing the bound is tight. 
If $\lambda=1$ and $t=1$, then Theorems \ref{thm1} and \ref{thm2} are consistent with the main results in \cite{chen1}. Hence, Our main results  
 generalize some of the results in \cite{chen1}.

\end{document}